\documentclass{article}
\usepackage[utf8]{inputenc}
\usepackage{geometry}
\usepackage[ruled,linesnumbered]{algorithm2e}
\usepackage{amsmath}
\usepackage{amsfonts}
\usepackage{amsthm}
\usepackage{mathtools}

\newtheorem{definition}{Definition}[section]
\newtheorem{lemma}{Lemma}[section]
\newtheorem{theorem}{Theorem}[section]

\newcommand\blfootnote[1]{%
  \begingroup
  \renewcommand\thefootnote{}\footnote{#1}%
  \addtocounter{footnote}{-1}%
  \endgroup
}

\begin{document}

\title{Local Routing in a Tree Metric 1-Spanner}

\author{Milutin Brankovic \and 
Joachim Gudmundsson\and 
Andr\'e van Renssen}

\date{}

\maketitle            
\begin{abstract}
 Solomon and Elkin \cite{tree-spanner} constructed a shortcutting scheme for weighted trees which results in a 1-spanner for the tree metric induced by the input tree. The spanner has logarithmic lightness, logarithmic diameter, a linear number of edges and bounded degree (provided the input tree has bounded degree). This spanner has been applied in a series of papers devoted to designing bounded degree, low-diameter, low-weight $(1+\epsilon)$-spanners in Euclidean and doubling metrics. In this paper, we present a simple local routing algorithm for this tree metric spanner. The algorithm has a routing ratio of 1, is guaranteed to terminate after $O(\log n)$ hops and requires $O(\Delta \log n)$ bits of storage per vertex where $\Delta$ is the maximum degree of the tree on which the spanner is constructed. This local routing algorithm can be adapted to a local routing algorithm for a doubling metric spanner which makes use of the shortcutting scheme.
 \blfootnote{A preliminary version of this paper appeared in the proceedings of COCOON'20 \cite{cocoon-tree-metric}.}

\end{abstract}
\section{Introduction}
Let $T$ be a weighted tree. The tree metric induced by $T$, denoted $M_T$, is the complete graph on the vertices of $T$ where the weight of each edge $(u, v)$ is the weight of the path connecting $u$ and $v$ in $T$. For $t\ge 1$, a $t$-spanner for a metric $(V, d)$ is a subgraph $H$ of the complete graph on $V$ such that every pair of distinct points $u, v\in V$ is connected by a path in $H$ of total weight at most $t\cdot d(u, v)$. We refer to such paths as $t$-spanner paths. A $t$-spanner has diameter $\Lambda$ if every pair of points is connected by a $t$-spanner path consisting of at most $\Lambda$ edges. Typically, $t$-spanners are designed to be sparse, often with a linear number of edges. The lightness of a graph is the ratio of its weight to the weight of its minimum spanning tree. Solomon and Elkin \cite{tree-spanner} define a 1-spanner for tree metrics. Given an $n$ vertex weighted tree of maximum degree $\Delta$, the 1-spanner has $O(n)$ edges, $O(\log n)$ diameter, $O(\log n)$ lightness and maximum degree bounded by $\Delta + O(1)$. While being an interesting construction in its own right, this tree metric 1-spanner has been used in a series of papers as a tool for reducing the diameter of various Euclidean and doubling metric spanner constructions \cite{short-thin-lanky,doubling-spanner,optimal-euclidean-spanners,tree-spanner,partitions-to-covers}.
 
 Once a spanner has been constructed, it becomes important to find these short paths efficiently. A local routing algorithm for a weighted graph $G$ is a method by which a message can be sent from any vertex in $G$ to a given destination vertex. The successor to each vertex $u$ on the path traversed by the routing algorithm must be determined using only knowledge of the destination vertex, the neighbourhood of $u$ and possibly some extra information stored at $u$. The efficiency of a routing algorithm is measured by the distance a message needs to travel through a network before reaching its destination as well as by the storage requirements for each vertex. There is a great deal of work on local routing algorithms in the literature. The difficulty of designing a good local routing algorithm clearly depends on the properties of the underlying network. Some authors have designed algorithms for very general classes of networks. For example, the algorithm of Chan et al. \cite{heirarchical-routing-in-doubling-metrics} works in any network although its quality depends on the induced doubling dimension. Many authors have focused on highly efficient algorithms for specific classes of networks. For example, there has been a line of research into routing algorithms for various classes of planar graphs \cite{convex-subdivisions,routing-in-delaunay}. Support for efficient local routing algorithms is a desirable feature in a spanner and in some recent papers, researchers have designed spanners which simultaneously achieve support for efficient local routing with other properties. See, for example, the work of Ashvinkumar et al. \cite{local-routing-in-sparse-and-light}. 
 
 There appears to be little work in the literature in designing spanners which achieve both support for local routing as well as low diameter. Given a spanner with diameter $\Lambda$, it is natural to require that a local routing algorithm on this spanner match this diameter, i.e., the algorithm should be guaranteed to terminate after at most $\Lambda$ hops. Abraham and Malkhi \cite{compact-routing-in-euclidean-metrics} give a construction, for any $\epsilon > 0$, of a $(1+\epsilon)$-spanner for points in two dimensional Euclidean space with an accompanying routing algorithm. The diameter of the routing algorithm is $O(\log D)$ with high probability where $D$ is the quotient of the largest and smallest inter-point distances. The routing algorithm has routing ratio $O(\log n)$, with high probability, for $n$ points on a uniform grid.
 
 In this paper, we demonstrate that the 1-spanner construction of Solomon and Elkin \cite{tree-spanner} supports a local routing algorithm with $O(\log n)$ diameter in the worst case. In Section~\ref{app:doubling}, we show that this routing algorithm can be adapted to a class of doubling metric $(1+\epsilon)$-spanners which employ the 1-spanner construction to achieve low diameter while retaining low weight and low degree.

 \section{The Model}
 
 A local routing algorithm for a weighted graph $G$ is a distributed algorithm in which each vertex is an independent processor. At the beginning of a round, a node may find that it has received a message. If a message is received, the algorithm decides to which neighbour the message should be forwarded. The following information is available to each vertex $v$ in $G$:
 
 \begin{enumerate}
     \item A header contained in the message.
     \item The label of $v$ as well as labels of neighbouring vertices.
 \end{enumerate}
 
 The message header stores the label of the destination vertex as well as other information if required by the routing algorithm. The labels are used not only as unique identifiers of vertices but may also store additional information if required. The labels of vertices are computed in a pre-processing step before the algorithm is run. Our model does not consider the running time of computation performed at each vertex in each round and so we do not specify the type of data structure used for headers and labels. The routing decision made at each vertex is deterministic. A routing algorithm is evaluated on the basis of the following quality measures.
 
 \begin{itemize}
     \item \textit{Routing Ratio}. Given two vertices $u$, $v$ of $G$, let $d_G(u, v)$ denote the shortest path distance from $u$ to $v$ and let $d_{route}(u, v)$ denote the total length of the path traversed by the routing algorithm when routing from $u$ to $v$. The routing ratio is defined to be $\max_{u,v \in V} \left\{ \frac{d_{route}(u, v)}{d_G(u, v)}\right\}$.
     
     \item \textit{Diameter}. A routing algorithm is said to have diameter $\Lambda$ if a message is guaranteed to reach its destination after traversing at most $\Lambda$ edges.
     
     \item \textit{Storage}. A bound on the number of bits stored at vertices and in message headers.
 \end{itemize}

\section{Local Routing in Tree Metrics}
In this section we present a slightly modified version of the tree metric 1-spanner of Solomon and Elkin \cite{tree-spanner}. We then present our routing algorithm for this spanner and show that it has a routing ratio of 1.
\subsection{The Spanner}
In this section we describe the tree metric 1-spanner construction of Solomon and Elkin \cite{tree-spanner}.

We first define some notation used in this section. $T$ will denote a weighted, rooted tree and $wt(T)$ will denote its weight. Given a graph $G$, $V(G)$ and $E(G)$ denote its vertex and edge sets respectively. The root of $T$ is denoted $rt(T)$ and $ch(v)$ denotes the number of children of $v$. As in the notation of \cite{tree-spanner}, the children of a vertex $v$ are denoted $c_1(v),...,c_{ch(v)}(v)$ and $p(v)$ denotes the parent of a vertex $v$. The lowest common ancestor of two vertices $u,v\in V(T)$ will be denoted by $lca(u, v)$. Given $u, v\in V(T)$, $P(u, v)$ denotes the unique path from $u$ to $v$ in $T$. Given $v\in V(T)$, $T_v$ will denote the subtree of $T$ rooted at $v$.

	The shortcutting procedure selects a constant number of cut vertices in the tree whose removal results in a forest of trees which are at most a constant fraction of the size of the input tree. The spanner is obtained by building the complete graph on these cut vertices and recursively applying the procedure to all subtrees obtained by removing the cut vertices. We note that the original construction appearing due to Solomon and Elkin \cite{tree-spanner} actually builds a low diameter spanner on the cut vertices rather than the complete graph.
	
	We first outline the method by which the cut vertices are selected. 
	We assume that among all subtrees rooted at children of a vertex $v$, the subtree rooted at the leftmost vertex is the largest. That is, $|T_{c_1(v)}| \ge |T_{c_i(v)}|$, for all $i > 1$.
	
	For an integer $d$, we call a vertex $v$ $d$-balanced if $|T_{c_1(v)}| \le |T| - d$.
	We label an edge $(u, v)$ of $T$ as leftmost if $u=c_1(v)$ or $v=c_1(u)$. Let $P(v)$ denote the path of maximum length from $v$ to some descendant of $v$ which includes only leftmost edges. We say the last vertex on $P(v)$ is the leftmost vertex in $T_v$ and we denote it by $l(v)$. We define $l(T):=l(rt(T))$. The construction we describe involves taking subtrees of an input tree. These subtrees inherit the `leftmost' labelling of the input tree so it may be the case that $l(T)=rt(T)$. If there is a $d$-balanced vertex along $P(v)$, we denote the first such vertex by $b_d(v)$. Otherwise, $b_d(v)=NULL$.

	Given a rooted tree $T$ and a positive integer $d$, we define a set of vertices $CV(T, d)$ as follows. Set $v:=rt(T)$. If $b_d(v)=NULL$, $CV(T, d)$ is defined to be $\emptyset$.
	Otherwise, $$CV(T, d):=\{b_d(v)\}\cup\left(\bigcup_{i=1}^{ch(b_d(v))}CV(T_{c_i(b_d(v))}, d)\right).$$
	
	Let $k$ be a positive integer. We define a set of vertices 
    $$
        C_{T} := 
        \begin{cases} 
                V(T) \mbox{ if } k \ge n/2 - 1, \\
                CV(T, d)\cup \{l(T), rt(T)\}  \mbox{ otherwise.}
        \end{cases}
    $$
    	
	where $d:=n/k$. (See Figure~\ref{fig:canonical-forest} for an example.) The spanner is constructed via the following recursive procedure which takes as input a tree $T$ and an integer parameter $k \ge 4$. Initialize the spanner as $G=T$. Compute the set $C_{T}$, with respect to $k$, and add the edges of the complete graph on $C_{T}$ to $G$. Denote by $T\setminus C_{T}$ the forest obtained by removing the vertices in $C_{T}$, along with all incident edges, from $T$. Recursively run the algorithm on all trees in the forest $T\setminus C_{T}$ and add the resulting edges to $G$. Note that the parameter $k$ is passed down to recursive calls of the algorithm while the parameter $d$ is recomputed based on the size of the subtree on which the algorithm is called. The following lemmas are established by Solomon and Elkin \cite{tree-spanner}:
	
	\begin{lemma}\label{size-of-cut-set}
	For $k\ge2$, the set $C_T$ contains at most $k+1$ vertices.
	\end{lemma}
	
	\begin{lemma}\label{elkin-solomon-lemma}
	Each tree in the forest $T\setminus C_T$ has size at most $2n/k$.
	\end{lemma}
	
	In particular, Lemma \ref{elkin-solomon-lemma} implies that the recursion depth of the spanner construction algorithm is $O(\log n)$ for $k\ge 4$.
	
	\begin{figure}
	    \centering
	    \includegraphics[width=30mm,scale=0.15]{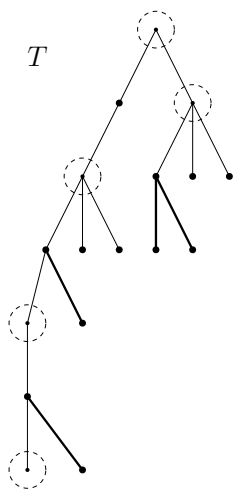}
	    \caption{A depiction of the set $C_T$ for a tree with $n=17$ vertices and the parameter $k$ set to $k=5$. The vertices of $C_T$ are inside the dashed circles. The vertices and edges of the forest $T\setminus C_{T}$  are shown in bold.} 
	    \label{fig:canonical-forest}
	\end{figure}
	
Solomon and Elkin~\cite{tree-spanner} show that the graph resulting from a slightly more elaborate version of this shortcutting scheme has weight $O(\log n)\cdot wt(T)$. Their scheme differs from what we have presented in that instead of building the complete graph on the set of cut vertices $C_T$, they build a certain 1-spanner with $O(k)$ edges and diameter $O(\alpha(k))$ where $\alpha$ is the inverse Ackermann function. Since we consider the parameter $k$ to be constant, this modification does not affect the weight bound of the construction. 
	
	\begin{theorem}[\cite{tree-spanner}]\label{spanner-theorem}
    	Let $T$ be a weighted rooted tree and let $k$ be a positive integer, $4\le k \le n/2 - 1$.
    	The graph $G$ obtained by applying the algorithm described above to $T$ using the parameter $k$ is a 1-spanner for the tree metric induced by $T$. Moreover, $G$ has diameter bounded by $O(\log_k n)$, weight bounded by $O(k^2\cdot \log_k n)\cdot wt(T)$ and maximum degree bounded by $\Delta + O(k)$ where $\Delta$ is the maximum degree of $T$.
	\end{theorem}
	
    Let $G$ be the spanner resulting from running the algorithm described above on some tree $T$ with respect to the parameter $k$. We define \textit{canonical subtrees} of $T$ with respect to $k$ to be the subtrees computed during the course of the construction of $G$. Canonical subtrees are defined recursively as follows. As a base for the recursive definition, the input tree $T$ is considered to be a canonical subtree with respect to $T$ and $k$. Suppose $T'$ is a canonical subtree with respect to $T$ and $k$. Then each tree in the forest $T'\setminus C_{T'}$ is a canonical subtree with respect to $T$ and $k$. We speak of canonical subtrees without reference to the parameters $T$ and $k$ when they are clear from the context. Given a vertex $v$ of $T$, we denote by $T^v$ the canonical subtree for which $v\in C_{T^v}$. When a canonical subtree $T'$ is small enough, $C_{T'}=V(T')$ and so it is clear that $C_{T^v}$ is well defined for each $v\in V(T)$. We say that $T^v$ is the canonical subtree of $v$ and that $v$ is a cut vertex of $T^v$.
    
	We establish some technical properties of canonical subtrees which will be of use in the following section. We reword statement 4 in Corollary 2.17 from the paper of Solomon and Elkin \cite{tree-spanner} in Lemma~\ref{at-most-two-entrances} and prove it using our terminology for completeness.
	
	\begin{lemma}\label{at-most-two-entrances}
		Let $T'$ be a canonical subtree of $T$ and let $T''$ be the canonical subtree such that $T' \in T''\setminus C_{T''}$. There are at most two edges in $T$ with one endpoint in $T'$ and the other outside $T'$. Moreover, any vertex of $T'$ incident to a vertex outside $T'$ must be $rt(T')$ or $l(T')$.
	\end{lemma}
	
	\begin{proof}
     We first need to establish the following claim.\\
     
     Claim:
     Let $u, v \in C_{T'}$ for some canonical subtree $T'$. Suppose $v$ is a descendant of $u'$ where $u'$ is some child of $u$. If $w$ is the first $d$-balanced vertex on the path from $u'$ to the left most vertex in $T'_{u'}$, then $v=w$ or $v$ is a descendant of $w$.\\
     
     Proof of claim:
     By definition of the cut vertex procedure, $v$ is a member of the set $CV(T'_{u'}, d)$. If $v=w$, we are done so suppose $v\ne w$. Every cut vertex in $CV(T'_{u'}, d)$ other than $w$ is contained in some $T'_{w'}$ for some child $w'$ of $w$. The claim follows.\\
     
     We note that any edge with one endpoint in $T'$ and the other outside $T'$ must have its other endpoint in $C_{T''}$. In particular, the parent of $x:=rt(T')$ must be an element of $C_{T''}$. Let $v_1$ be the parent of $x$. Suppose there is a second edge from a vertex $v_2 \in C_{T''}$ to a vertex $w$ in $V(T')$. Note that $v_2$ must be a child of $w$ for otherwise $T$ would contain a cycle. Then $v_2$ is a descendant of $x$. Let $x'$ be the first $d$-balanced vertex on the path from $x$ to the left most vertex in $T''_{x}$. Note that $x'\in C_{T''}$. By the claim, $v_2$ either coincides with, or is a descendant of, $x'$. If the latter holds, the parent of $v_2$ cannot be a vertex in $T'$ and so $v_2=x'$. Suppose there exists a third vertex $v_3\in C_{T''}$ which has a parent in $T'$. Then $v_3$ is a member of $C_{T''}$ which is a descendant of $x$ and neither equal to, nor a descendant of, the vertex $x'$ which is impossible given the claim. Since $v_2$ is the first $d$-balanced vertex on the path from $x$ to the left most vertex in $T''_x$, $v_2$ must be the leftmost child of its parent. It follows that the parent of $v_2$ is $l(T')$.
 \end{proof}

	\begin{lemma}\label{root-of-canonical-subtree}
	    Let $T'$ and $T''$ be canonical subtrees such that $T''\in T'\setminus C_{T'}$ and let $v$ be some vertex of $T''$. Let $X$ be the set of vertices in $C_{T'}$ which are ancestors of $v$. Let $x$ be the element of $X$ deepest in $T$ and let $x'$ be the child of $x$ which is an ancestor of $v$. Then $x'=rt(T'')$. (See Figure \ref{fig:root-of-canonical-subtree}.)
	\end{lemma}
	
\begin{proof}
    Consider the path $P(rt(T'), v)$. Note that all vertices in $X$ lie on this path. It is easy to see that this path enters $T''$ through $rt(T'')$. The vertex on this path preceding $rt(T'')$ is a member of $C_{T'}$, and an ancestor of $v$ and therefore a member of $X$. Since the predecessor of $rt(T'')$ is clearly deeper than all other vertices in $X$, the lemma follows.
 \end{proof}
	
	\begin{figure}
	    \centering
	    \includegraphics[width=50mm,scale=0.5]{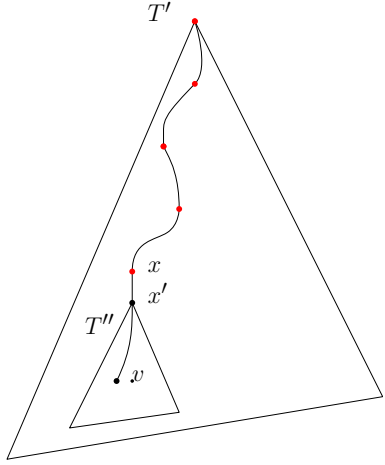}
	    \caption{Lemma \ref{root-of-canonical-subtree}. The set $X$ consists of the red vertices.} 
	    \label{fig:root-of-canonical-subtree}
	\end{figure}
	
	In the spanner, a vertex $u$ is connected to all vertices in $C_{T^u}$. The following lemma ensures that in the routing algorithm described in the next section, under certain conditions, it is safe to make a `greedy' choice from a subset of vertices in $C_{T^u}$.
	
	\begin{lemma}\label{u-tree-in-v-tree-u-ancestor}
	Let $u$ and $v$ be vertices in $G$ such that $v$ is not a vertex of $T^u$. Let $X$ be the set of vertices in $C_{T^u}$ which are ancestors of $v$. Suppose $X\ne \emptyset$ and let $x$ be the last vertex on the path from $u$ to $v$ which is contained in $T^u$. Then $x\in X$. Moreover, $x$ is the deepest vertex in $X$. (See Figure \ref{fig:u-tree-in-v-tree-u-ancestor}.)
	\end{lemma}
	
    \begin{proof}
      By Lemma \ref{at-most-two-entrances}, $x$ is either $rt(T^u)$ or $l(T^u)$. Since $X\ne \emptyset$, $rt(T^u)$ is an ancestor $v$ which also implies $l(T^u)$ is an ancestor of $v$. Then $x$ is an ancestor of $v$ and it remains to be shown that $x$ is the deepest vertex in $X$. Let $x^*$ be the deepest vertex in $X$ and suppose $x^* \ne x$. The path from $x^*$ to $v$ must leave $T^u$ through either $rt(T^u)$ or $l(T^u)$. Since we also have that $x\in \{rt(T^u), l(T^u)\}$ and since $x^*\ne x$, we must have $x=rt(T^u)$ and $x^*=l(T^u)$. Then there would be two   distinct paths from $rt(T^u)$ to $v$, a contradiction since   $T$ is a tree. The lemma follows.
    \end{proof}
	
	\begin{figure}[h!]
	    \centering
	    \includegraphics[width=50mm,scale=0.5]{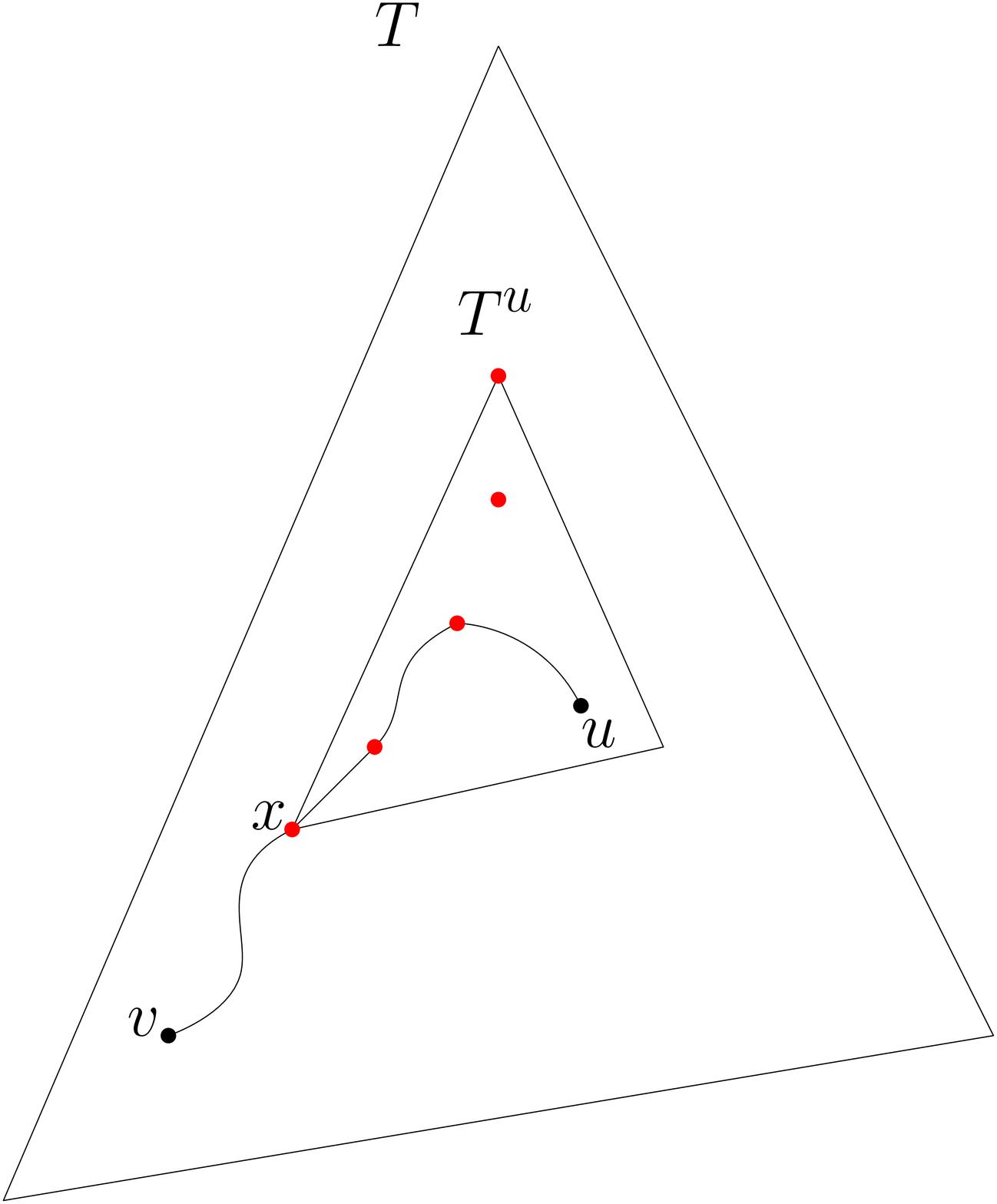}
	    \caption{Lemma \ref{u-tree-in-v-tree-u-ancestor}. The set $X$ consists of the red vertices.} 
	    \label{fig:u-tree-in-v-tree-u-ancestor}
	\end{figure}
    
    \begin{lemma}\label{u-descendant-of-v}
        Let $u$ and $v$ be vertices of $T$ such that $u$ is a descendant of $v$ and $T^v$ is contained in $T^u$. Let $T'$ be the canonical subtree in the forest $T^u\setminus C_{T^u}$ which contains $v$ and let $X$ be the set of vertices in $C_{T^u}$ which are descendants of $v$ and ancestors of $u$. Let $x$ be the element of $X$ which is highest in $T$. Then either $rt(T')$ or $l(T')$ is the parent of $x$. (See Figure \ref{fig:u-descendant-of-v}.)
    \end{lemma}
    
    \begin{proof}
      Note that $X=P(u, v)\cap C_{T^u}$. Let $w$ be the first vertex on $P(u, v)$ which is contained in $T'$. By Lemma \ref{at-most-two-entrances}, $w\in \{rt(T'), l(T')\}$. Let $x'$ be the predecessor of $w$ on $P(u, v)$.  Since $v$ is an ancestor of $u$, $x'$ is a child of $w$. Since $x'$ is a vertex outside $T'$ connected to a vertex in $T'$, $x'\in C_{T^u}$ and hence $x'\in X$. Clearly $x'$ is also the highest vertex in $X$ and so $x'=x$ and the lemma follows.
    \end{proof}
    
    \begin{figure}[h!]
	    \centering
	    \includegraphics[width=50mm,scale=0.5]{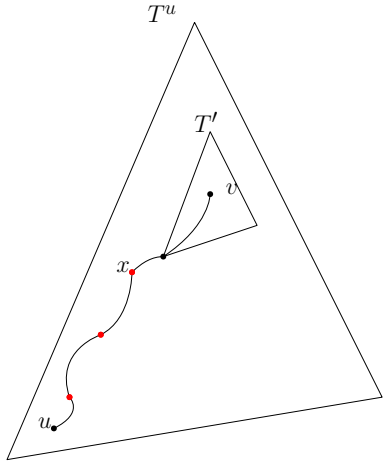}
	    \caption{Lemma \ref{u-descendant-of-v}. The set $X$ consists of the red vertices.} 
	    \label{fig:u-descendant-of-v}
	\end{figure}

\subsection{Routing Algorithm}
	In this section, we describe a local routing algorithm for the spanner defined above. Throughout what follows, let $G$ denote the graph obtained when the algorithm of the previous section is applied to a weighted, rooted tree $T$ using a parameter $k\ge 4$.
	We first define the labels $label(v)$ for vertices $v\in V(G)$.
	We make use of the interval labelling scheme of Santaro and Khatib \cite{santoro-khatib}.
    Let $rank(v)$ denote the rank of $v$ in a post-order traversal of $T$. We define $$L(v):=\min\{rank(w):w\in V(T_v)\}.$$
    We define the label of $v$ to be $label(v)=[L(v), rank(v)]$. The observation used in the routing algorithm of Santaro and Khatib \cite{santoro-khatib} is that a vertex $w$ is a descendant of $v$ if and only if $rank(w)\in[L(v), rank(v)]$. Note that the label of each vertex can be computed in linear time in a single traversal of the tree.
	
	\begin{lemma}\label{small-data}
	    In the labelling scheme outlined above, each  vertex of $G$ stores $O((\Delta + k)\log n)$ bits of information.
	\end{lemma}
	\begin{proof}
	    Since $G$ has $n$ vertices, for each $v\in V(G)$, $rank(v)$ and $L(v)$ are both integers in the interval $[1,...,n]$ and therefore require $O(\log n)$ bits to be represented. By Theorem \ref{spanner-theorem}, each vertex of $G$ has at most $\Delta + O(1)$ neighbours in $G$. Then, for any $v\in V(G)$, we require $O((\Delta + k)\log n)$ bits to store $rank(w)$ and $L(w)$ for each neighbour $w$ of $v$.
	\end{proof}
	
	For our routing algorithm, no auxiliary data structure is required at each vertex and so the total storage requirement per vertex is $O((\Delta + k)\log n)$ bits.
    Our routing algorithm considers a number of cases depending on the labels of the current vertex $u$ and destination vertex $v$.
    For convenience of analysis, in each case we specify two routing steps. For ease of exposition, we consider a vertex to be both a descendant and ancestor of itself.
	
	\textbf{Case 0:} If $v$ is a neighbour of $u$, route to $v$.\\
	
	\textbf{Case 1:}
	$u$ is an ancestor of $v$ in $T$. Let $X$ be the set of vertices in $C_{T^u}$ which are ancestors of $v$. Let $x$ be the deepest element of $X$. Route first to $x$ and then to the child of $x$ which is an ancestor of $u$. \\

	\textbf{Case 2:}
	$u$ is a descendant of $v$ in $T$. Let $X$ be the set of vertices in $C_{T^u}$ which are descendants of $v$ and ancestors of $u$. Let $x$ be the highest vertex in $X$. Route first to $X$ and then to its parent.\\

	\textbf{Case 3:} $u$ is not an ancestor or descendant of $v$. Let $X$ be the set of vertices in $C_{T^u}$ which are ancestors of $v$ and not ancestors of $u$. If $X\ne \emptyset$, we define $x$ to be the deepest vertex in $X$ and define $x'$ to be the child of $x$ which is an ancestor of $v$. Let $Y$ be the set of vertices in $C_{T^u}$ which are ancestors of $u$ but not ancestors of $v$. We define $y$ to be the highest vertex in $Y$.\\ 
	
	\indent \textbf{Case 3 a):} $X$ is empty. Route first to $y$ and then to the parent of $y$.\\
	\\
	\indent \textbf{Case 3 b):} $X$ is non-empty. Route first to $x$ and then to $x'$.\\
	
	The routing algorithm uses a greedy strategy. Given the destination vertex and neighbours due to tree edges and shortcut edges, the algorithm simply selects the edge that appears to make the most progress. It is not obvious that this strategy gives the desired $O(\log n)$ diameter of the routing algorithm. Indeed, this bound would not hold if the shortcuts were arbitrary edges and hinges on the particular structure of the 1-spanner.
	
	The arguments we make in our analysis will make use of certain integer sequences we assign to vertices of $G$. Note that these sequences are used only for the analysis of the algorithm and are not part of the labelling scheme. We define a unique integer sequence for each canonical subtree computed during the course of the spanner construction. Each vertex $v$ will be assigned the sequence corresponding to the tree $T^v$.
	
	 The integer sequence for each canonical subtree is defined recursively as follows. The input tree $T$ is given the empty sequence. Suppose $T'$ is a canonical subtree which has already been associated with some sequence $S$. Consider the forest $T'\setminus C_{T'} = \{T_1,...,T_p\}$. Each tree $T_j \in T'\setminus C_{T'}$ is associated with the sequence obtained by appending $j$ to $S$. It is clear that every vertex of $T$ appears in $C_{T'}$ for exactly one canonical subtree $T'$. Let $S_v$ denote the sequence assigned to the vertex $v$. We refer to $S_v$ as the canonical sequence of $v$.
	 
	Observe that if for two vertices $u, v\in V(G)$ we have that $S_u = S_v$, by definition of the spanner construction algorithm, $u$ and $v$ must be cut vertices of the same canonical subtree of $T$ and are therefore connected by an edge in $G$. The routing algorithm works by choosing a successor vertex so as to incrementally transform the canonical sequence of the current vertex into the canonical sequence of the destination vertex.

	\begin{lemma}\label{always deeper after case 1}
    	Let $u$ and $v$ be vertices of $G$ such that $u$ is an ancestor of $v$. Consider the two vertices visited after executing the routing steps of Case 1 when routing from $u$ to $v$. These vertices are on the path from $u$ to $v$ in $T$. Moreover, these vertices are visited in the order they appear on this path.
	\end{lemma}
	
	\begin{lemma}\label{always higher after case 2}
	    Let $u$ and $v$ be vertices of $G$ such that $u$ is a descendant of $v$. Consider the two vertices visited after executing the routing steps of Case 2 when routing from $u$ to $v$. These vertices are on the path from $u$ to $v$ in $T$. Moreover, these vertices are visited in the order they appear on this path.
	\end{lemma}
	
	Lemma \ref{always deeper after case 1} and Lemma \ref{always higher after case 2} are immediate from the definition of the routing steps.
	
	\begin{lemma}\label{case-3-a)-lemma}
	    Let $u$ and $v$ be vertices of $G$ such that $u$ is not an ancestor or a descendant of $v$. Suppose the set $X$, as defined in Case 3 of the routing algorithm, is empty. Consider the vertices visited after executing Case 3 a) of the routing algorithm. These vertices are on the path from $u$ to the $v$. Moreover, these vertices are visited by the routing algorithm in the order they appear on this path.
	\end{lemma}
	\begin{proof}
	    Let $w_1$ be the first vertex visited in Case 3 a) and let $w_2$ be the second. Then $w_1$ is the vertex $y$ and $w_2$ is its parent. Since $y$ is, by definition, not an ancestor of $v$, it lies on the path from $u$ to $lca(u, v)$ and is not equal to $lca(u, v)$. Since $w_2$ is the parent of $y$, it is clearly the next vertex on $P(u, lca(u, v))$.
	\end{proof}
	
	\begin{lemma}\label{case-3-b)-lemma}
	    Let $u$ and $v$ be vertices of $G$ such that $u$ is not an ancestor or descendant of $v$. Suppose the set $X$, as defined in Case 3 of the routing algorithm, is non-empty. Consider the vertices visited after executing Case 3 b) of the routing algorithm. These vertices are on the path from $lca(u, v)$ to $v$. Moreover, these vertices are visited in the order they appear on this path.
	\end{lemma}
	\begin{proof}
	    Let $w_1$ be the first vertex visited in Case 3 b) and let $w_2$ be the second. $w_1$ is an ancestor of $v$ which is not an ancestor of $u$. It follows that $w_1$ is a descendant of $lca(u,v)$. By definition of Case 3 b), $w_2$ is the child of $w_1$ which is an ancestor of $v$. It is clear that $w_2$ is the successor to $w_1$ on the path $P(lca(u, v), v)$. The lemma follows.
	\end{proof}
	
	Lemmas \ref{always deeper after case 1}, \ref{always higher after case 2}, \ref{case-3-a)-lemma} and \ref{case-3-b)-lemma} imply the following:
	
	\begin{lemma}\label{edge-on-path}
    	Let $u$ and $v$ be vertices of $G$ and let $P(u, v)=(u=x_1,...,x_p=v)$. Suppose the routing steps of a single case of the routing algorithm are executed and vertices $w_1$ and $w_2$ are visited. Then there are indices $1\le i_1\le i_2 \le p$ such that $w_1 = x_{i_1}$ and $w_2=x_{i_2}$.
	\end{lemma}
	
	Using Lemma \ref{edge-on-path}, we can show the routing algorithm has a routing ratio of 1.
	
	\begin{theorem}\label{routing-ratio}
	    Let $u$ and $v$ be vertices of $G$. Let $\delta_{T}(u, v)$ denote the length of the path from $u$ to $v$ in $T$. The routing algorithm described above is guaranteed to terminate after a finite number of steps and the length of the path traversed is exactly $\delta_{T}(u, v)$.
	\end{theorem}
	
	\begin{proof}
         Lemma \ref{edge-on-path} implies that the vertices visited when any case of the routing algorithm is executed are on the path from the current vertex to the destination and they are explored in the order they appear on the path. Each case leads to at least one new vertex being explored. If $P=(u=x_1,x_2,...,x_p=v)$ is the path from $u$ to $v$ in $T$, the arguments made above imply that the path traversed by the routing algorithm is of the form $P'=(x_{i_1}, x_{i_2},...,x_{i_l})$ where $1\le i_1<i_2<...<i_l\le p$, i.e., $P'$ is a sub-sequence of $P$. In the tree metric $M_{T}$, the weight of the edge $(x_{i_j}, x_{i_{j+1}})$ is by definition equal to the weight of the path $(x_{i_j}, x_{i_j+1}, ..., x_{i_{j+1}})$ and so the theorem follows.
    \end{proof}

	We now argue that the routing algorithm is guaranteed to terminate after traversing $O(\log n)$ edges. To this end, we first prove the following lemma.

	\begin{lemma}\label{argument-about-canonical-sequences-in-case-1-and-case-2}
 	    Let $u$ and $v$ be vertices of $T$ such that $u$ is either an ancestor or a descendant of $v$. Let $u'$ be the vertex reached after executing the routing steps of either Case 1 or Case 2 when routing from $u$ to $v$. Then the following statements hold:
	    \begin{enumerate}
	        \item If $S_u$ is a prefix of $S_v$, then $|S_{u'}| > |S_u|$. Moreover, either $S_{u'} = S_v$ or $S_{u}'$ is a prefix of $S_v$.
	        \item If $S_v$ is a prefix of $S_u$, then $|S_{u'}| < |S_u|$. Moreover, either $S_{u'} = S_v$ or $S_v$ is a prefix of $S_{u'}$.
	        \item Suppose $S_u$ and $S_v$ share a common prefix $S$ of length $m < \min\{|S_u|, |S_v|\}$. Then $|S_{u'}| < |S_u|$. Moreover, either $S_{u'} = S$ or $S$ is a prefix of $S_{u'}$
	    \end{enumerate}
	\end{lemma}
	\begin{proof}
	    We begin with the first statement. Suppose $S_u$ is a prefix of $S_v$. Then $T^u$ contains $T^v$. Let $T'$ be the canonical subtree in the forest $T^u\setminus C_{T^u}$ which contains $T^v$. By definition, the canonical sequence of any cut vertex of $T'$ can be obtained by appending some integer to $S_u$. Since $T^v$ is contained in $T'$, the canonical sequence associated to $T'$ is a prefix of $S_v$. Then it is sufficient to show that $u'\in C_{T'}$. Suppose $u$ is an ancestor of $v$ so that the algorithm executes the routing steps of Case 1. Recall that in Case 1 the set $X$ is defined as the set of vertices in $C_{T^u}$ which are ancestors of $v$ and the vertex $x$ is defined as the deepest vertex in $X$. Then by definition of Case 1, $u'$ is the child of $x$ which is an ancestor of $v$. By Lemma \ref{root-of-canonical-subtree}, $u'=rt(T')$. Since $rt(T')\in C_{T'}$, the statement of the lemma follows in this case. Now suppose $u$ is a descendant of $v$ so that the algorithm executes the routing steps of Case 2. Recall that in Case 2, $X$ is defined to be the set of vertices in $C_{T^u}$ which are descendants of $v$ and ancestors of $u$ and $x$ is defined to be the highest vertex in $X$. By definition of Case 2, the algorithm routes to $x$ and then to the parent of $x$. Then $u'$ is the parent of $x$. By Lemma \ref{u-descendant-of-v}, $u'\in \{l(T'), rt(T')\}$. Since $\{l(T'), rt(T')\}\subseteq C_{T'}$, the first statement of the lemma holds in this case.
	    
	    We now address the second and third statements of the lemma. Suppose $S_u$ is not a prefix of $S_v$. Let $S$ be the longest common prefix of $S_u$ and $S_v$. Then either $S=S_v$ or $S$ is a prefix of $S_v$. Let $T''$ be the canonical subtree corresponding to the canonical sequence $S$. Observe that $T''$ contains $T^u$. Let $T'$ be the canonical subtree such that $T^u\in T'\setminus C_{T'}$. Then either $T''=T'$ or $T''$ contains $T'$. Note that the canonical sequence of any cut vertex of $T'$ is a prefix of $S_u$. Moreover, for any canonical sequence $S'$ of a cut vertex in $T'$, either $S=S'$ or $S$ is a prefix of $S'$. We claim that $u'\in C_{T'}$. When $S_v$ is a prefix of $S_u$, we see this implies the second statement. When $S_u$ and $S_v$ share a prefix of length $m < \min\{|S_u|, |S_v|\}$, we see that our claim implies the third statement.
	    
	    We now prove the claim. Suppose that $u$ is an ancestor of $v$ so that the algorithm executes the routing steps of Case 1. Let $X$ and $x$ be as defined in Case 1. Then $u'$ is the child of $x$ which is an ancestor of $v$. By Lemma \ref{u-tree-in-v-tree-u-ancestor}, $x$ is the last vertex on the path from $u$ to $v$ which is contained in $T$. Since $x$ is an ancestor of $v$ and $u'$ is both a child of $x$ and ancestor of $v$, it is clear that $u'$ is the next vertex on $P(u, v)$. Since $u$ is a vertex outside $T'$ connected to a vertex in $T^u$, we see that $u'\in C_{T'}$ and so the claim holds in this case. Suppose now that $u$ is a descendant of $v$ so that the algorithm executes the routing steps of Case 2. Let $X$ and $x$ be as defined in Case 2. Note that since $u$ is a descendant of $v$, $rt(T^u)$ must also be a descendant of $v$. Then $rt(T^u)\in X$. Since $rt(T^u)$ must be the highest vertex in $X$, we see that $x=rt(T^u)$. Since the parent of $rt(T^u)$ is clearly a member of $C_{T'}$, the claim holds. This completes the proof of the lemma. 
	    
	\end{proof}
	
	Note that if $u$ is an ancestor (resp. descendant) of $v$, then by Lemma \ref{edge-on-path} the vertex reached after executing the steps of Case 1 (resp. Case 2) will also be an ancestor (resp. descendant) of $v$.
	
	\begin{lemma}\label{logn-steps-when-above-or-below}
	    Suppose $u$ and $v$ in $G$ are such that $u$ is an ancestor or descendant of $v$ in $T$. Then, when routing from $u$ to $v$, the routing algorithm reaches $v$ after traversing $O(\log n)$ edges.
	\end{lemma}
	
	\begin{proof} 
     Suppose that $u$ is an ancestor of $v$ (the argument in the case where $u$ is a descendant of $v$ is symmetric) and let $K$ be the maximum length of the canonical sequence of any vertex of $G$. Note that by Lemma \ref{elkin-solomon-lemma}, each canonical subtree is at most $2/k$ times the size of its smallest containing canonical subtree and so the recursion depth is at most $O(\log n)$ which, by definition of canonical sequences, implies $K=O(\log n)$. If $S_u$ is a prefix of $S_v$, by Lemma \ref{argument-about-canonical-sequences-in-case-1-and-case-2}, the routing steps executed in a single iteration of Case 1 traverse at most two edges and lead to a vertex $u'$ which is an ancestor of $v$ and for which $S_{u'}$ is a prefix of $S_v$ strictly longer than $S_u$. Then after traversing at most $2\cdot K$ edges, the routing algorithm will reach $v$. An analogous argument can be made in the case where $S_v$ is a prefix of $S_u$.
     \\
     Assume that $S_u$ and $S_v$ have a common prefix of length $m < \min\{|S_u|, |S_v|\}$. By Lemma \ref{argument-about-canonical-sequences-in-case-1-and-case-2}, after traversing at most $2\cdot K$ edges, the algorithm reaches a vertex $u'$ which is an ancestor of $v$ and is such that $S_{u'}$ is a prefix of $S_v$. By the argument made above, after traversing at most another $2\cdot K$ edges, the routing algorithm reaches $v$. Then in all cases, after traversing at most $4\cdot K$ edges, the algorithm reaches its destination. This completes the proof of the lemma.
    \end{proof}

	Consider the case where $u$ is neither an ancestor nor a descendant of $v$. The following lemma shows that in this case, the algorithm either routes to a vertex on the path $P(lca(u, v), v)$ or it follows the routing steps that would be executed if the algorithm were routing from $u$ to $lca(u, v)$.
	
	\begin{lemma}\label{case-2-simulation}
	    Let $u$ and $v$ be vertices of $G$ such that $lca(u, v)\notin \{u, v\}$. Suppose that the set $X$ as defined in Case 3 is empty so that the algorithm executes the routing steps of Case 3 a) when routing from $u$ to $v$. The same steps would be performed when routing from $u$ to $lca(u, v)$. 
	\end{lemma}
	
	 \begin{proof}
         Recall the set $Y$ in Case 3 is defined as the set of ancestors of $u$ which are not ancestors of $v$. Consider Case 2 when routing from $u$ to $lca(u, v)$. The set $X$ in Case 2 is defined as the set of ancestors of $u$ which are descendants of $lca(u, v)$. Since an ancestor of $u$ is a descendant of $lca(u, v)$ if and only if it is not an ancestor of $v$, we see that $X=Y$. In Case 3 a) the algorithm routes to the highest vertex in $Y$ and then its parent. In Case 2 the algorithm routes to the highest vertex in $X$ and then its parent. We see the algorithm executes the same routing steps in both cases and so the lemma follows.
    \end{proof}
	
	Using Lemmas \ref{logn-steps-when-above-or-below} and \ref{case-2-simulation}, we establish the main result of this section.
	
	\begin{theorem}\label{main-theorem}
	    Let $u$ and $v$ be vertices in $G$. The routing algorithm reaches $v$ when routing from $u$ after traversing at most $O(\log n)$ edges.
	\end{theorem}
	
\begin{proof}
     Lemma \ref{logn-steps-when-above-or-below} implies the theorem when $u$ is an ancestor or descendant of $v$. Then we need only consider the case where $u$ is neither an ancestor nor descendant of $v$. Suppose the algorithm executes the steps of Case 3 a) at some point when routing from $u$ to $v$. Let $u'$ be the current vertex after the steps of Case 3 a) are executed. Then $u'$ is an ancestor of $v$ and so the algorithm will reach $v$ after traversing another $O(\log n)$ edges. It remains to argue that $O(\log n)$ edges are traversed before the current vertex is an ancestor of $v$. By Lemma \ref{case-2-simulation} and Lemma \ref{logn-steps-when-above-or-below}, the algorithm either executes the steps of Case 3 a) or reaches $lca(u, v)$ after traversing $O(\log n)$ edges. This completes the proof of the theorem.
 \end{proof}

\section{Routing in a Doubling Metric Spanner}
\label{app:doubling}
The 1-spanner of the previous section was designed to be a tool for reducing the diameter of Euclidean spanners without compromising bounds on degree and weight. It has been applied in a series of papers on various constructions of spanners of doubling metrics. The doubling dimension of a metric $(M, d)$ is the smallest number $\lambda$ such that, for any $r > 0$, any ball of radius $r$ can be covered by at most $2^\lambda$ balls of radius $r/2$. An $n$-point metric is said to be doubling if the doubling dimension is constant in $n$. The class of doubling metrics contains the class of Euclidean metrics as $m$-dimensional Euclidean metrics can be shown to have doubling dimension $\Theta(m)$. The spanners for doubling and Euclidean metrics to which the 1-spanner shortcutting scheme has been applied are all based on some variety of hierarchical decomposition of the input point set. In particular, spanners based on the dumbbell trees of Arya et al. \cite{short-thin-lanky} and the net tree approach \cite{def-spanner-and-app,fast-construction-of-nets} have been proposed. We will show that the routing algorithm described in the previous section can be extended to a routing algorithm on a low diameter doubling metric $(1+\epsilon)$-spanner based on the net tree.

\subsection{A net tree based spanner}

The spanner we describe here is a simplified version of the spanner described by Chan et al. \cite{doubling-spanner}. The starting point of their construction is based on the spanner of Gao et al. \cite{def-spanner-and-app}.\\

In this section, $(M, d)$ will denote an $n$-point doubling metric. Let $\dim(M)$ denote the doubling dimension of $M$. We assume without loss of generality that the distance between the closest pair of points in $M$ is 1.

\begin{definition}
	For $r>0$, a subset $N\subseteq M$ is said to be an $r$-net for $M$ if \begin{itemize}
		\item $d(x_1, x_2) > r$ for all distinct pairs of points $x_1, x_2\in N$.
		\item For all $y\in M$, there exists $x\in N$ such that $d(x, y)\le r$.
	\end{itemize} 
\end{definition}

The first step of the construction is to compute a certain sequence of nets. Specifically, we compute a nested sequence of sets $M=N_0\supseteq N_1\supseteq ... \supseteq N_{\log D}$ where $N_i$ is a $2^i$-net of $N_{i-1}$ and $D$ is the largest inter-point distance among pairs of points in $M$. Note that $N_{\log D}$ consists of a single point. An efficient method for computing such a sequence of nets is given by Har-Peled and Mendel \cite{fast-construction-of-nets}. Next, we describe the so-called net tree which we denote $\mathcal{T}$. The net tree has a node at level-$i$ for each point in the set $N_i$. If a level-$i$ node corresponds to a point $p\in N_i$, we say that $p$ is the representative of $v$ and write $rep(v)=p$. Note that a given point in $p$ can be the representative of multiple nodes in $\mathcal{T}$. The root of the tree corresponds to the singleton set $N_{\log D}$ and the leaves of $\mathcal{T}$ are in one-to-one correspondence with the points of $M$. The parent of each node $v$ at level-$i$ is a node $w$ at level-($i+1$) for which $d(rep(v), rep(w)) \le 2^{i+1}$. The constant doubling dimension of $M$ immediately implies that nodes in $\mathcal{T}$ have a constant number of children. The tree induces a graph $H$ on $M$ in the natural way, i.e., each edge $(u, v)\in E(\mathcal{T})$ corresponds to the edge $(rep(u), rep(v))\in E(H)$. For a point $p\in M$, we also use $p$ to denote the leaf of $\mathcal{T}$ which has $p$ as its representative. Given a leaf $p$, we denote by $p^{(i)}$ the level-$i$ ancestor of $p$ in $\mathcal{T}$. Given two nodes $u, v\in V(\mathcal{T})$, we define $d(u, v):=d(rep(u), rep(v))$. The graph $H$ as currently defined is not a $(1+\epsilon)$-spanner for $M$. To guarantee a $1+\epsilon$ spanning ratio, certain cross edges are added at each level. Fix a constant $\gamma > 4$. For  each pair of level-$i$ points $u, v\in V(\mathcal{T})$ for which $d(rep(u), rep(v))\le \gamma\cdot2^i$, we add the cross edge $(u, v)$ to $\mathcal{T}$. If we make a sufficiently large choice of $\gamma = O(1/ \epsilon)$, the graph $H$ resulting from the addition of these cross edges to $\mathcal{T}$ can be shown to be a $(1+\epsilon)$-spanner for $M$. The following theorem is established in the proof of Theorem 3.2 in~\cite{def-spanner-and-app}.

\begin{theorem}\label{structure-of-spanning-paths}
	Let $p, q$ be points in $M$ and let $i$ be the smallest integer such that the level-$i$ ancestors of $p$ and $q$ are connected by a cross edge. Consider the path in $\mathcal{T}$ obtained by climbing from $p$ to its level-$i$ ancestor, taking the cross edge to the level-$i$ ancestor of $q$ and descending to $q$. Then the path in $H$ corresponding to this path has total length $(1+\epsilon)\cdot d(p, q)$.
\end{theorem}

The maximum degree of the spanner as described above is $O(\epsilon^{-\dim(M)}\log D)$ as some points in $M$ may represent many nodes in $\mathcal{T}$. Gottlieb and Roddity~\cite{improved-algorithms-for-fully-dynamic} devised a technique to reduce the degree bound to $O(\epsilon^{-dim(M)})$ by carefully reassigning representatives.

The depth of the net tree is $O(\log D)$. In particular, if $D$ is not bounded by a polynomial in $n$, the net tree may have linear depth and so the diameter of the spanner may be linear. We say a subtree of $\mathcal{T}$ is \textit{light} if it has height $\log(\frac{D}{n})$. Chan et al. \cite{doubling-spanner} show that by applying the tree metric construction of the previous section to all light subtrees of $\mathcal{T}$, it is possible to reduce the diameter of the spanner to $O(\log n)$ without increasing the asymptotic bound on the weight. Chan et al. \cite{doubling-spanner} show that the weight of the spanner, after shortcuts are added, is at most $O(\epsilon^{-dim(M)}\log n)\cdot wt(MST(M))$, where $MST(M)$ is a minimum spanning tree of $M$. Chan et al. \cite{doubling-spanner} make additional modifications to ensure fault tolerance although we will not work with this version of the spanner. 

To summarise, the spanner we work with has $(1+\epsilon)$-stretch, $O(\epsilon^{-\dim(M)})$ degree, $O(\epsilon^{-dim(M)}\log n)\cdot wt(MST(M))$ weight and $O(\log n)$ diameter.

For the remainder of this section, $\mathcal{T}^*$ will denote the net tree with cross edges and shortcuts. $H$ denotes the constant degree spanner of $(M, d)$ induced by $\mathcal{T}^*$.
Our routing algorithm will make use of the following property of the spanner~$H$.

\begin{lemma}\label{interval-lemma}
	For $p, q\in M$, the smallest integer $i$ such that the level-$i$ ancestors of $p$ and $q$ are guaranteed to be joined by a cross edge lies in the interval  $$\left[\left\lfloor \log\left(\frac{d(p, q)}{\gamma + 4}\right) \right\rfloor, \left\lceil \log\left(\frac{d(p, q)}{\gamma - 4}\right) \right\rceil \right].$$
\end{lemma}

\begin{proof}
	Lemma 4.1 in the paper of Gao et al. \cite{def-spanner-and-app} implies that if $(p^{(i)}, q^{(i)})$ is a cross edge, then $(p^{(j)}, q^{(j)})$ is a cross edge for all $j\ge i$. Then $i$ is the unique integer for which $(p^{(i)}, q^{(i)})$ is a cross edge and $(p^{(i-1)}, q^{(i-1)})$ is not. Observe that for $j<\log\left( \frac{d(p, q)}{\gamma + 4} \right)$, we have $d(p, q) > (\gamma + 4)\cdot2^j$.
	By definition of the net tree, $d(p^{(j)}, p) \le \sum_{k=1}^{j}2^k \leq 2\cdot2^j$. By the triangle inequality we have $d(p, q) \le d(p^{(j)}, q^{(j)}) + d(p,p^{(j)}) + d(q, q^{(j)}) \leq d(p^{(j)}, q^{(j)}) + 4\cdot2^j$ and so $d(p^{(j)}, q^{(j)}) \ge d(p, q) - 4\cdot2^j > \gamma\cdot2^j.$ Then by the construction of the spanner, there is no cross edge between $p^{(j)}$  and $q^{(j)}$. A similar argument shows that for $j > \left\lceil \log\left(\frac{d(p, q)}{\gamma - 4}\right)  \right\rceil$, $d(p^{(j)}, q^{(j)}) \le \gamma \cdot 2^j$ and so there is guaranteed to be a cross edge between $p^{(j)}$ and $q^{(j)}$. Then since $(p^{(j)}, q^{(j)})$ is not a cross edge for $j=\left\lfloor \log\left(\frac{d(p, q)}{\gamma + 4}\right) \right\rfloor-1$ and is guaranteed to be a cross edge for $j=\left\lceil \log\left(\frac{d(p, q)}{\gamma - 4}\right) \right\rceil$, the lemma follows.
\end{proof}

Note that the length of the interval computed in Lemma \ref{interval-lemma} is $\log\left( \frac{\gamma+4}{\gamma-4} \right) = O(1)$. This fact is needed to establish a bound on the diameter of the routing algorithm we present in the next section.

\subsection{Labelling Scheme and Routing Algorithm}
Gottlieb and Roddity \cite{improved-algorithms-for-fully-dynamic} show it is possible to route in the net tree in their construction while taking care to use the first available cross edge. This algorithm has a routing ratio of $1+\epsilon$ and the required labels are short as a result of the degree bound. We apply their method to the version of the spanner have described although we make use of our 1-spanner routing algorithm whenever the message passes through short subtrees so as to take advantage of the available shortcuts.
We first describe the labels of $\mathcal{T}^{*}$. Since each point in $p\in M$ is the representative of multiple nodes in $\mathcal{T}^{*}$, the label of $p$ will contain the labels of all nodes $v\in V(\mathcal{T}^{*})$ for which $p=rep(v)$. Note that there is a constant number of such nodes in $\mathcal{T}^{*}$ for each $p\in M$.

Initially, each node $v\in V(\mathcal{T}^{*})$ is labelled with the interval $[L(v), rank(v)]$, the same labelling scheme used in our 1-spanner routing algorithm. Next, for each light subtree $\mathcal{T}'$ of $\mathcal{T}$, we label the vertices of $\mathcal{T}'$ with the labels required for the 1-spanner routing algorithm on $\mathcal{T}'$.

When routing from a point $p$ to a point $q$ in $M$, in order to obtain a $(1+\epsilon)$ routing ratio, the algorithm must make use of the first cross edge connecting ancestors of $p$ and $q$. Since we make use of shortcuts, there is a danger the algorithm may `overshoot' the correct level. Therefore, the algorithm requires a means to estimate the level of the tree containing a viable cross edge, i.e., the algorithm should be able to locally compute the interval of Lemma \ref{interval-lemma}. If $M$ is a Euclidean metric, the distance $d(p, q)$ can be computed if the labels of the points store their coordinates. However, if $M$ is an arbitrary doubling metric, it is not obvious how to achieve this goal. For $\delta > 0$, a $(1+\delta)$-ADLS is a labeling of points in a metric space that allows approximation of pairwise distances to within a factor of $1+\delta$. That is, given the labels of points $x$ and $y$ in a $(1+\delta)$-ADLS, we are able to compute a value $\tilde{d}(x, y)$ such that $$d(x, y) \le \tilde{d}(x, y) \le (1+\delta)d(x, y).$$ Har-Peled and Mendel \cite{fast-construction-of-nets} give a $(1+\delta)$-ADLS for doubling metrics in which each label consists of at most $\min\{\delta^{-O(\dim(M))}\cdot\log D, \delta^{-O(\dim(M))}\cdot \log n(\log n + \log\log D)\}$ bits. If we work with a $(1+\delta)$-ADLS rather than exact distances, the interval of Lemma~\ref{interval-lemma} becomes $$\left[\left\lfloor \log\left(\frac{\tilde{d}(p, q)}{(\gamma + 4)(1+\delta)}\right) \right\rfloor, \left\lceil \log\left(\frac{\tilde{d}(p, q)}{(\gamma - 4)}\right) \right\rceil \right].$$ Note that the length of this interval is still $O(1)$.

To summarize, the label of each point $p\in M$ contains:
\begin{itemize}
	\item The label of $p$ in a $(1+\delta)$-ADLS for $M$.
	\item The label of each node $v\in V(\mathcal{T}^{*})$ for which $p=rep(v)$.
\end{itemize}
Since the degree of the spanner is bounded by $O(\epsilon^{-dim(M)})$, we obtain the following bound on the storage of our labelling scheme.
\begin{lemma}\label{size-of-doubling-spanner-labels}
	In the labelling scheme for the spanner $H$ defined above, each label requires at most $O(\epsilon^{-dim(M)}\log n) + \eta$ bits of storage, where $\eta = \min\{\delta^{-O(dim(M))}\cdot\log D,  \delta^{-O(dim(M))}\cdot \log n (\log n + \log\log D)\}$.
\end{lemma}

We now describe the routing algorithm for the spanner $H$. Let $p$ be the current vertex and $q$ be the destination. Since each point of $H$ is the representative of multiple vertices in $\mathcal{T}^{*}$, the message header stores a $O(\log n)$ bit variable to keep track of the current position in $\mathcal{T}^{*}$. Suppose the current position is a vertex $v\in \mathcal{T}^{*}$ with representative $p$. When we say the algorithm routes to a neighbour $w$ of $v$, we mean the message is forwarded to a neighbour $p'$ of $p$ in $H$ such that $rep(w)=p'$. Throughout the following, $u$ will denote the current position of the algorithm in $\mathcal{T}^{*}$ and $v$ will denote the leaf in $\mathcal{T}^{*}$ corresponding to the destination vertex in~$H$.\\
The algorithm operates in the following distinct states.

\textbf{Initialization:} Using the ADLS labels stored in the labels of $p$ and $q$, a $(1+\delta)$-approximation of $d(p, q)$ is computed. Let $\tilde{d}(p, q)$ denote this approximation. The integer $i:=\left\lfloor \log\left(\frac{\tilde{d}(p, q)}{(1 + \delta)(\gamma + 4)}\right) \right\rfloor$ is stored in the message header as the cross edge target level.\\

\textbf{Ascending: } In this state, the current position $u$ of the message in $\mathcal{T}^{*}$ is at a level lower than $i$. Suppose $u$ is in a light subtree. Then we use our 1-spanner routing algorithm to route towards the level-$i$ ancestor of $u$. Otherwise, we simply route to the parent of $u$.\\

\textbf{Searching for a cross edge: } In this state, $u$ is at level $i$ or higher in $\mathcal{T}^{*}$. If $u$ is connected by a cross edge to an ancestor of $v$, the algorithm routes to this ancestor. Otherwise, the algorithm routes to the parent of $u$.\\

\textbf{Descending: } In this state, $u$ is an ancestor of $v$ in $\mathcal{T}^{*}$. If $u$ is in a light subtree, the algorithm routes towards $v$ using our 1-spanner routing algorithm. Otherwise, the algorithm routes to the child of $u$ which is an ancestor of $v$.

\begin{theorem}
	The routing ratio of the algorithm defined above has routing ratio $1+\epsilon$ and diameter $O(\log n)$.
\end{theorem}
\begin{proof}
	Observe that the path traversed when routing from $p$ to $q$ by the algorithm differs from the $(1+\epsilon)$ spanning path described in Theorem \ref{structure-of-spanning-paths} only by the shortcuts taken on sections of the path which pass through light subtrees. Since these shortcuts cannot increase the length of the path, the first statement of the theorem follows.
	
	To prove the second statement, we bound the number of vertices visited in each state. In the ascending state, at most $O(\log n)$ vertices are visited in a light subtree since the 1-spanner routing algorithm has $O(\log n)$ diameter. Since there are at most $O(\log n)$ levels of the net tree higher than level $\log(\frac{D}{n})$, at most $O(\log n)$ vertices are visited by the algorithm outside a light subtree. The same arguments hold for the descending state.
	Note that the level at which to begin the search for a cross edge determined in the initialization state is at most a constant number of levels lower that the lower endpoint of the interval of Lemma~\ref{interval-lemma}. Since the length of the interval in Lemma \ref{interval-lemma} has constant length, the second statement of the theorem follows.
\end{proof}

\section{Conclusion}
We have shown that a simplified version of the logarithmic diameter tree metric 1-spanner of Solomon and Elkin \cite{tree-spanner} supports a local routing algorithm of routing ratio 1 and logarithmic diameter. We have also shown that this algorithm can be used to route in a spanner for doubling metrics which uses the shortcutting scheme to achieve a low diameter while maintaining low weight and low degree.

\bibliographystyle{plain}
\bibliography{references}
\end{document}